\documentclass[11pt]{article}

\usepackage{geometry}
\usepackage{amsmath}
\usepackage{amssymb}
\usepackage{graphicx}
\usepackage{parskip}
\usepackage{pstricks,pst-node}
\usepackage{graphicx}
\usepackage{amsthm}  

\newtheorem{theorem}{Theorem}
\newtheorem{lemma}{Lemma}

\newcommand{\Exp}{\mathrm{Exp}}
\newcommand{\U}{\mathrm{U}}

\newcommand{\dx}{\mathrm{d}x}

\begin{document}


\title{Multi-probe consistent hashing}
\author{Ben Appleton, Michael O'Reilly \\
  \small{\{appleton, moreil\}@google.com} \\
  \emph{Google}}
\maketitle

\section{Abstract}

We describe a consistent hashing algorithm which performs multiple lookups per
key in a hash table of nodes. It requires no additional storage beyond the hash
table, and achieves a peak-to-average load ratio of $1 + \varepsilon$ with just
$1 + \frac{1}{\varepsilon}$ lookups per key.

\newpage

\section{Introduction}

Consistent hashing was introduced by Karger et al. \cite{Karger1997}. It allows
multiple clients to balance keys over a set of nodes without communication, and
continue to agree for almost every key as the collection of live nodes changes
over time or varies across machines. First applied to caching, consistent
hashing has also been applied to key-value stores such as Dynamo
\cite{DeCandia2007}, and routing tables such as Kademlia \cite{Maymounkov2002},
used notably by BitTorrent.

A figure of merit for load balancing is the \emph{peak-to-average load
  ratio}. In this paper we consider the case where there are many distinct keys
per node, and define the load on a node as the proportion of keys which map to a
node. Then peak-to-average load ratio is the ratio of the maximum load to the
average load over the set of nodes.

In addition to consistent placement, an ideal consistent hash might have the
following performance properties\footnote{Here $O(\ldots)$ should be interpreted
  liberally to allow in expectation, with high probability, amortized, or
  \emph{etc}.}:
\begin{enumerate}
  \item $O(n)$ space for $n$ nodes; ideally just the collection of nodes.
  \item $O(1)$ time per insertion or removal.
  \item $O(1)$ time per lookup.
  \item Peak-to-average load ratio at most $1 + \varepsilon$, for some small
    $\varepsilon$.
\end{enumerate}
However existing algorithms fall short of this ideal.

Karger et al's \emph{ring consistent hash} \cite{Karger1997} hashes each node
$O(\frac{\ln n}{\varepsilon^2})$ ways to a ring, indexing each node hash. To
assign a key to a node, it hashes the key to the ring and returns the node with
the next hash. However to obtain a peak-to-average load ratio of $1 +
\varepsilon$ it requires $O(\frac{n \ln n}{\varepsilon^2})$ memory.

Thaler and Ravishankar's highest random weight algorithm \cite{Thaler1998}
hashes each key against each node, returning the node with the highest resulting
hash. For a large number of keys this produces a peak-to-average load ratio of
1, but takes $O(n)$ time per lookup. Wang and Ravishankar \cite{Wang2009}
present a variation which takes $O(\ln n)$ time, by clustering nodes into a
pre-agreed tree then recursively selecting clusters by hashing the key against
each cluster down the tree. However this requires pre-agreement on the
hierarchy, with no provision for handling changes to the set of live nodes.

Lamping and Veach's \emph{jump consistent hash} \cite{Lamping2014} hashes each
key to a list of nodes labeled $1, 2, \ldots, n$. Keys are placed using a
pseudo-random number generator to compute a sequence of node assignments as the
number of nodes grows. It takes $O(1)$ space and $O(\ln n)$ expected time,
achieving a peak-to-average load ratio of exactly $1$. However, it does not
support the removal of arbitrary nodes. This prevents from using jump consistent
hash in applications which must handle arbitrary node loss. It also prevents
from using jump consistent hash in weighted consistent hashing
\cite{Schindelhauer2005}.

\section{Analysis}

In this paper we define `with high probability' to be with probability $1 -
\frac{1}{n^{\Omega(1)}}$.

\subsection{Hashing keys and nodes once each - naive approach}

Consider hashing $n$ nodes to the unit ring. When presented with a key, hash the
key to the unit ring, and return the next node along the ring. This requires
$O(n)$ memory and $O(1)$ time per lookup, but produces a poor peak-to-average
load.

Let $X_1, \ldots, X_n$ be the distances between successive node hashes, such
that the probability of selecting node $i$ is $X_i$. For a node hash function
selected at random from a universal ranged hash family, the distances are
distributed with $\Pr(X_i > x_i) = (1 - x_i)^n$. As Feller shows
(\cite{Feller1971} Chapter I.7), in the limit $n \rightarrow \infty$ the
distances converge to a Poisson process with independent and identically
distributed (iid) $X_i \sim \Exp(n)$; a simplifying approximation which we use
subsequently. Therefore with high probability $\max_{i=1}^n X_i =
\Theta(\frac{\ln n}{n})$. Since the mean load is $\frac{1}{n}$, the
peak-to-average load ratio is then $\Theta(\ln n)$. Since a service must be
provisioned for peak load, but its capacity is proportional to the average load,
a high peak-to-average load ratio may be unacceptable in many applications.

\subsection{Hashing nodes $J$ ways}

Ring consistent hash \cite{Karger1997} resolves this load imbalance by hashing
each node $J = O(\ln n)$ ways to \emph{virtual nodes} on the unit ring. The
virtual nodes are stored in a hash table. When presented with a key, hash the
key to the unit ring, find the next virtual node, then return the corresponding
physical node.

If we use $J$ independent hashes, the set of node hashes forms a Poisson process
with iid $X_{i, j} \sim \Exp(J n)$, and fraction of keys assigned to node $j$ is
$S_j = \sum_{i = 1}^J X_{i, j}$ with mean $\frac{1}{n}$. By Cramer's theorem
\cite{Cramer1938} $P(S_j > \frac{1 + \varepsilon}{n}) \sim e^{-J I(1 +
  \varepsilon)}$ where $I(t) = J n t - 1 - \ln(J n t)$ for this process. So to
achieve a peak-to-average load ratio of $1 + \varepsilon$ in expectation or with
high probability requires $J = \Theta(\frac{\ln n}{\varepsilon^2})$.

Note that $J$ cannot be changed online as that would break consistency. So $J$
must be sized for the maximum number of nodes expected in the lifetime of the
system, or provision must be made for changes that break consistency.

\subsection{Hashing keys $K$ ways}

We propose to store each node once but to hash keys $K$ ways, returning the
subsequent node which is closest to a key hash. We call this \emph{multi-probe
  consistent hashing}. This requires $O(n)$ space to store $n$ nodes and $O(K)$
time per lookup. Perhaps surprisingly, using only $K = 2$ key hashes improves
the peak-to-average load ratio from $O(\ln n)$ to $O(1)$ with high probability.

The key result is Theorem \ref{theorem:peak-to-average}. For technical reasons
we begin by deriving the expected peak-to-average load ratio.

\begin{lemma} \label{lemma:peak-to-average}
For $K$ independent hashes per key with $2 \leq K \ll \frac{\sqrt{n}}{\ln n}$,
for a random node hash function from a universal ranged hash family, the peak
load is $\frac{K}{K-1}\frac{1}{n} + o\left(\frac{1}{n}\right)$ in expectation.
\end{lemma}

\begin{proof}
Consider a $K+1$-independent universal ranged hash family, selecting $1$ node
hash function and $K$ key hash functions. Then the node hashes form a Poisson
process with rate $n$. Without loss of generality let $x_1 = \max_{i=1}^n x_i$,
such that node $1$ is the maximally loaded node. For $K \ll \frac{\sqrt{n}}{\ln
  n}$ the probability that multiple key hashes resolve to $x_1$ is
$o\left(\frac{1}{n}\right)$, which case we will neglect. Else if 1 key hash
resolves to $x_1$ it has distance $\sim \U(0, x_1)$ and the $K - 1$ other key
hashes have iid distances $\sim \Exp(n + 1)$, where the latter is obtained by
considering the key hash as another node hash. Then the probability that a key
is assigned to $x_1$ is:
\begin{equation}
\begin{split}
& K \int_{x=0}^1 e^{-(n+1)(K-1)x} \dx + o\left(\frac{1}{n}\right) \\
& = \frac{K}{K-1}\frac{1}{n} + o\left(\frac{1}{n}\right)
\end{split}
\end{equation}
\end{proof}

In \cite{McDiarmid1998} McDiarmid proved:
\begin{lemma} \label{lemma:mcdiarmid}
Let $X_1, \ldots, X_n$ be a family of independent random variables. Suppose that
the real-valued function $Z$ satisfies
\begin{equation}
\left|Z(x) - Z(x')\right| \leq c_k
\end{equation}
whenever the vectors $x$ and $x'$ differ only in the $k$th coordinate. Let
$\mu$ be the expected value of the random variable $Z(X)$. Then for any $\lambda
\geq 0$,
\begin{equation}
\Pr(\left|Z(X) - \mu\right| \geq \lambda \sigma) \leq 2e^{-2\lambda^2}
\end{equation}
where $\sigma^2 = \sum_{k=1}^n c_k^2$.
\end{lemma}
We proceed to use McDiarmid's inequality to prove that the bound in Lemma
\ref{lemma:peak-to-average} holds not just in expectation but with high
probability.

\begin{theorem} \label{theorem:peak-to-average}
For $K$ independent hashes per key with $2 \leq K \ll \frac{\sqrt{n}}{\ln n}$,
for a random node hash function from a universal ranged hash family, the peak
load is $\frac{K}{K-1}\frac{1}{n} + o\left(\frac{1}{n}\right)$ with high
probability.
\end{theorem}

\begin{proof}
For one independent key hash, the probability that the distance to the next node
hash is at most $x$ is $F(x)$, with:
\begin{equation}
1 - F(x) = \sum_{i=1}^n
\begin{cases}
x_i - x & \text{if } x \leq x_i \\
0 & \text{if } x > x_i \\
\end{cases}
\end{equation}

For $K \ll \frac{n}{\ln n}$ the probability that a single key hashes multiple
times to the maximally loaded node is $o(\frac{1}{n})$. Then defining
\begin{equation} \label{equation:peak-load}
Z_K = K \int_{x=0}^1 \left(1 - F(x)\right)^{K-1} \dx
\end{equation}
the peak load for $K$ key hashes is $Z_K + o(\frac{1}{n})$. For a random node
hash function, the expected value of $Z_K$ is $\frac{K}{K-1}\frac{1}{n} +
o(\frac{1}{n})$.

Recall we have $\mu = \frac{K}{K-1}\frac{1}{n} + o\left(\frac{1}{n}\right)$, and
$0 \leq x_k \leq \frac{c \ln n}{n}$ with high probability.

Begin with the case $K = 2$. Equation \ref{equation:peak-load} simplifies to
$Z_2 = \sum_{i=1}^n x_i^2$, which gives $c_k = \left(\frac{c \ln n}{n}\right)^2$
and hence $\sigma = \frac{\left(c \ln n\right)^2}{n \sqrt{n}} =
o\left(\frac{1}{n}\right)$, so $Z_2 = \frac{2}{n} + o\left(\frac{1}{n}\right)$
with high probability.

For the case $K > 2$ we obtain
\begin{equation}
\begin{split}
c_k & = K \frac{c \ln n}{n} (K-1) \int_{x=0}^1 (1 - F(x))^{K-2} \dx \\
 & = K \frac{c \ln n}{n} Z_{K-1} \\
 & = O\left(\frac{K \ln n}{n^2}\right) \\
\end{split}
\end{equation}
The first equality notes that $x_k$ appears $K - 1$ times in $(1-F(x))^{K-1}$,
whose difference with respect to $x_k$ is at most $\frac{c \ln n}{n}(K - 1)(1 -
F(x))^{K-2}$, discarding higher-order terms by $K = o(\frac{n}{\ln n})$. We then
induct on $K$. Then by $\sigma^2 = \sum_{k=1}^n c_k^2$ we have $\sigma =
O\left(\frac{K \ln n}{n \sqrt{n}}\right)$. Since $K = o(\frac{\sqrt{n}}{\ln n})$
we obtain $\sigma = o(\frac{1}{n})$, and hence the desired result.
\end{proof}

So the peak-to-average load ratio is $\frac{K}{K-1} + o(1)$ with high
probability. To achieve a peak-to-average load ratio of $1 + \varepsilon$
requires $K = 1 + \frac{1}{\varepsilon}$ hashes, and $O(\frac{1}{\varepsilon})$
time per lookup.

\subsection{Other properties}

Karger et al \cite{Karger1997} defined other important properties for consistent
hash functions: \emph{monotonicity}, \emph{spread} and \emph{load}. For
completeness we will address these properties. Our proofs closely mirror
\cite{Karger1997}.

\begin{theorem}
  For $K = O(1)$ the hash family $F$ described in this paper has the following
  properties:
\begin{enumerate}
  \item Monotonicity: $F$ is monotone.
  \item Spread: If the number of views $V = \rho n$ for some constant $\rho$,
    and the number of keys $I = n$, then for $i \in \mathcal{I}$, $\sigma(i)$
    is $O(t \ln n)$ with high probability.
  \item Load: If $V$ and $I$ are as above, then for $n \in \mathcal{N}$,
    $\lambda(n)$ is $O(t \ln n)$ with high probability.
\end{enumerate}
\end{theorem}

\begin{proof}
Monotonicity: Adding a node to the ring does not increase the distance from any
key hash to the next node, and does not reduce the distance from any key hash to
any existing node. So no key can switch to an existing node.

Spread and load follow from the observation that with high probability, a point
from every view falls into an interval of length $O(\frac{t \ln n}{n})$. Spread
follows by observing that for each key, the number of node points that fall
within this size interval around the $K$ key hashes, $O(t \ln n)$, is an upper
bound on the spread of that key. Load follows by counting the number of key
hashes that fall in the region owned by a node hash, $O(t \ln n)$.
\end{proof}

\subsection{Performance summary}

\begin{table}
  \centering
  \caption{Properties of each algorithm}
  \label{table:performance}
  \begin{tabular}{|r || r | r | r |}
    \hline
     & Jump c. h. & Ring c. h. & Multi-probe c. h. \\
    \hline
    Peak-to-average & $1$ & $1 + \varepsilon$ & $1 + \varepsilon$ \\
    Memory & $O(1)$ & $O(\frac{n \ln n}{\varepsilon^2})$ & $O(n)$ \\
    Update time & 0 & $O(\frac{\ln n}{\varepsilon^2})$ & $O(1)$ \\
    Assignment time & $O(\ln n)$ & $O(1)$ & $O(\frac{1}{\varepsilon})$ \\
    Arbitrary node removal & No & Yes & Yes \\
    \hline
  \end{tabular}
\end{table}

Table \ref{table:performance} summarizes the performance of each algorithm for
$n$ nodes.

\section{Implementation}

For a hash table we use an array of sorted, inlined vectors. The array is sized
to about 6 nodes per inlined vector. The inlined vector stores the first 8
elements inline, then spills to an out-of-line buffer. This avoids
pointer-chasing in the common case. We use $64$-bit identifiers and hashes. We
store the hash alongside the node identifier to save on subsequent hashes.

\section{Performance}

We compare multi-probe consistent hash to ring consistent hash \cite{Karger1997}
and jump consistent hash \cite{Lamping2014}.

\subsection{Peak-to-average load ratio}

We measured the peak-to-average load ratio over a range of node counts. For each
node count we ran 1,000 trials using different node hash seeds to obtain
percentiles over the statistic of interest: peak-to-average load ratio. For each
trial we sampled 1,000,000 keys per node. These simulations were run on a
cluster of machines.

\begin{table}
  \centering
  \caption{Peak-to-average, $K = 2$}
  \label{table:peak-to-average:K=2}
  \begin{tabular}{|r || r | r | r|}
    \hline
    Number of nodes & median of trials & 90\%ile of trials & 99\%ile of trials \\
    \hline
    10 & 1.74 & 2.43 & 3.32 \\
    100 & 1.96 & 2.22 & 2.48 \\
    1,000 & 2.00 & 2.08 & 2.16 \\
    10,000 & 2.00 & 2.03 & 2.05 \\
    100,000 & 2.00 & 2.01 & 2.02 \\
    \hline
  \end{tabular}
\end{table}

Table \ref{table:peak-to-average:K=2} shows the peak-to-average load ratio for
multi-probe consistent hash with $K = 2$. The peak-to-average load ratio
converges to 2. This requires 30-60 ns per lookup and 2.2MB of memory for the
largest set of nodes.

\begin{table}
  \centering
  \caption{Peak-to-average, $K = 21$}
  \label{table:peak-to-average:K=21}
  \begin{tabular}{|r || r | r | r|}
    \hline
    Number of nodes & median of trials & 90\%ile of trials & 99\%ile of trials \\
    \hline
    10 & 1.04 & 1.13 & 1.24 \\
    100 & 1.05 & 1.08 & 1.10 \\
    1,000 & 1.05 & 1.06 & 1.07 \\
    10,000 & 1.05 & 1.06 & 1.06 \\
    100,000 & 1.05 & 1.06 & 1.06 \\
    \hline
  \end{tabular}
\end{table}

Table \ref{table:peak-to-average:K=21} shows the peak-to-average load ratio for
multi-probe consistent hash with $K = 21$. The peak-to-average load ratio
converges to 1.05.

\begin{table}
  \centering
  \caption{Peak-to-average, $J = \ln n$}
  \label{table:peak-to-average:J=1}
  \begin{tabular}{|r | r || r | r | r|}
    \hline
    Number of nodes & J & median of trials & 90\%ile of trials & 99\%ile of trials \\
    \hline
    10 & 2 & 2.23 & 3.05 & 3.96 \\
    100 & 4 & 2.64 & 3.24 & 4.05 \\
    1,000 & 6 & 2.84 & 3.29 & 3.75 \\
    10,000 & 9 & 2.79 & 3.11 & 3.51 \\
    100,000 & 11 & 2.89 & 3.15 & 3.40 \\
    \hline
  \end{tabular}
\end{table}

Table \ref{table:peak-to-average:J=1} shows the peak-to-average load ratio for
ring consistent hash with $J = \ln n$ hashes per node. The peak-to-average load
ratio converges to $e$.

\begin{table}
  \centering
  \caption{Peak-to-average, $J = 700 \ln n$}
  \label{table:peak-to-average:J=700}
  \begin{tabular}{|r | r || r | r | r|}
    \hline
    Number of nodes & J & median of trials & 90\%ile of trials & 99\%ile of trials \\
    \hline
    10 & 1611 & 1.04 & 1.06 & 1.08 \\
    100 & 3223 & 1.05 & 1.06 & 1.07 \\
    1,000 & 4835 & 1.05 & 1.05 & 1.06 \\
    10,000 & 6447 & 1.05 & 1.05 & 1.06 \\
    100,000 & - & - & - & - \\
    \hline
  \end{tabular}
\end{table}

Table \ref{table:peak-to-average:J=700} shows the peak-to-average load ratio for
ring consistent hash with $J = 700 \ln n$ hashes per node. The peak-to-average
load ratio converges to 1.05. At 10,000 nodes the table required 1,400MB of
memory. At 100,000 nodes the table did not fit in memory on the available
machines.

\subsection{Timings}

For multi-probe consistent hash we set $K = 21$, obtaining a peak-to-average
load ratio of 1.05.

For ring consistent hash we set $J = 700 \ln n$, obtaining a peak-to-average
load ratio of 1.05.  The implementation of ring consistent hash uses a hash
table for $O(1)$ assignment, similar to the implementation of multi-probe
consistent hash.

The implementation of jump consistent hash is taken without modification from
\cite{Lamping2014}. Jump consistent hash achieves a peak-to-average load ratio
of 1.0.

All implementations are in C++. Binaries are compiled on a 64-bit platform using
GNU C++ and measured on 1 core of an Intel Xeon W3690 @3.47GHz.

\begin{table}
  \centering
  \caption{Initialization time (ns) per node}
  \label{table:initialization}
  \begin{tabular}{|r || r | r | r|}
    \hline
    Number of nodes & Multi-probe c. h. & Ring c. h. & Jump c. h. \\
     & $\varepsilon = 0.05$ & $\varepsilon = 0.05$ & $\varepsilon = 0$ \\
    \hline
    10 & 28 & 58,000 & 0 \\
    100 & 29 & 175,000 & 0 \\ 
    1,000 & 31 & 555,000 & 0 \\ 
    10,000 & 41 & 910,000 & 0 \\ 
    100,000 & 40 & - & 0 \\ 
    \hline
  \end{tabular}
\end{table}

Table \ref{table:initialization} shows the initialization time per
node. Multi-probe consistent hash is constant except for a step at 10,000
nodes, as the hash table spills to L3 cache\footnote{Cache spilling is visible
  throughout the timings. We will not comment upon each instance.}. Ring
consistent hash requires orders of magnitude more initialization time per
node. Jump consistent hash requires no initialization.

\begin{table}
  \centering
  \caption{Memory (bytes) per node}
  \label{table:memory}
  \begin{tabular}{|r || r | r | r|}
    \hline
    Number of nodes & Multi-probe c. h. & Ring c. h. & Jump c. h. \\
     & $\varepsilon = 0.05$ & $\varepsilon = 0.05$ & $\varepsilon = 0$ \\
    \hline
    10 & 22 & 35,000 & 0 \\
    100 & 22 & 71,000 & 0 \\ 
    1,000 & 22 & 106,000 & 0 \\ 
    10,000 & 22 & 142,000 & 0 \\ 
    \hline
  \end{tabular}
\end{table}

Table \ref{table:memory} shows the memory per node, where we have used 64 bit
hashes and 64 bit node identifiers. Multi-probe consistent hash uses constant
memory per node. Ring consistent hash requires orders of magnitude more memory,
commensurate with its high initialization time. Jump consistent hash requires
no memory.

\begin{table}
  \centering
  \caption{Assignment time (ns)}
  \label{table:assignment}
  \begin{tabular}{|r || r | r | r|}
    \hline
    Number of nodes & Multi-probe c. h. & Ring c. h. & Jump c. h. \\
     & $\varepsilon = 0.05$ & $\varepsilon = 0.05$ & $\varepsilon = 0$ \\
    \hline
    10 & 350 & 29 & 32 \\
    100 & 420 & 60 & 50 \\ 
    1,000 & 430 & 110 & 67 \\ 
    10,000 & 590 & 130 & 80 \\ 
    100,000 & 590 & - & 94 \\ 
    \hline
  \end{tabular}
\end{table}

Table \ref{table:assignment} shows the time per key. Multi-probe consistent
hash is takes constant time modulo cache effects. Ring consistent hash is a few
times faster. Jump consistent hash is generally fastest as it does not access
memory.

\begin{table}
  \centering
  \caption{Update time (ns)}
  \label{table:update}
  \begin{tabular}{|r || r | r | r|}
    \hline
    Number of nodes & Multi-probe c. h. & Ring c. h. & Jump c. h. \\
     & $\varepsilon = 0.05$ & $\varepsilon = 0.05$ & $\varepsilon = 0$ \\
    \hline
    10 &      33 &   135,000 & 0 \\
    100 &     51 &   360,000 & 0 \\
    1000 &    70 & 1,000,000 & 0 \\
    10000 &   79 & 1,800,000 & 0 \\
    100000 & 107 &         - & 0 \\
    \hline
  \end{tabular}
\end{table}

Table \ref{table:update} shows the amortized time per insertion or removal of a
node, measured by inserting from empty to full then removing from full to empty
again (in random order). Multi-probe consistent hash requires only $O(1)$
amortized time per insertion or removal. Ring consistent hash requires orders of
magnitude more time per update. Jump consistent hash requires no time to update,
as it does not maintain a hash table.

It's important to note that all timings above are for uncontended caches, such
that the hash table of nodes are cached near the CPU. However caches are
typically contended in practice, which may evict the hash table of nodes to L3
or even main memory. Key assignment and node updates may be commensurately
slower for multi-probe and ring consistent hash.

\section{Discussion}

Jump consistent hash is not generally applicable, as it cannot handle the loss
of an arbitrary node. However where applicable it generally requires less time
and space than the alternatives, in which case we recommend jump consistent
hash.

Ring consistent hash has fast key assignment: just one hash table
lookup. However to achieve a peak-to-average load ratio of $1 + \varepsilon$
over $n$ nodes it requires $O(\frac{n \ln n}{\varepsilon^2})$ memory,
potentially multiple gigabytes in practice. It is correspondingly slow to
initialize and to update.

Multi-probe consistent hash stores each node just once in a hash table, so it
requires only $O(n)$ memory and supports updates in $O(1)$ expected amortized
time. To achieve a peak-to-average load ratio of $1 + \varepsilon$ it requires
$O(\frac{1}{\varepsilon})$ time per lookup. In practice it can achieve a
peak-to-average load ratio of 1.05 in 350-600 ns per key assignment, while
scaling to larger node sets than possible with ring consistent hash. This makes
multi-probe consistent hash an attractive replacement for ring consistent hash.

It's interesting to note the similarity between multi-probe consistent hash and
cuckoo hashing \cite{Pagh2004}, in which hashing keys two ways achieves a load
factor up to $\frac{1}{2}$ for an in-memory hash table. The authors speculate
that there might be fruitful connections to explore here.

\section{Acknowledgements}

The authors would like to thank Eric Veach, Raul Vera and David Symonds for many
insightful comments on the draft paper.

\bibliographystyle{plain}
\bibliography{paper}

\newpage

\begin{figure}
  \centering
  \includegraphics{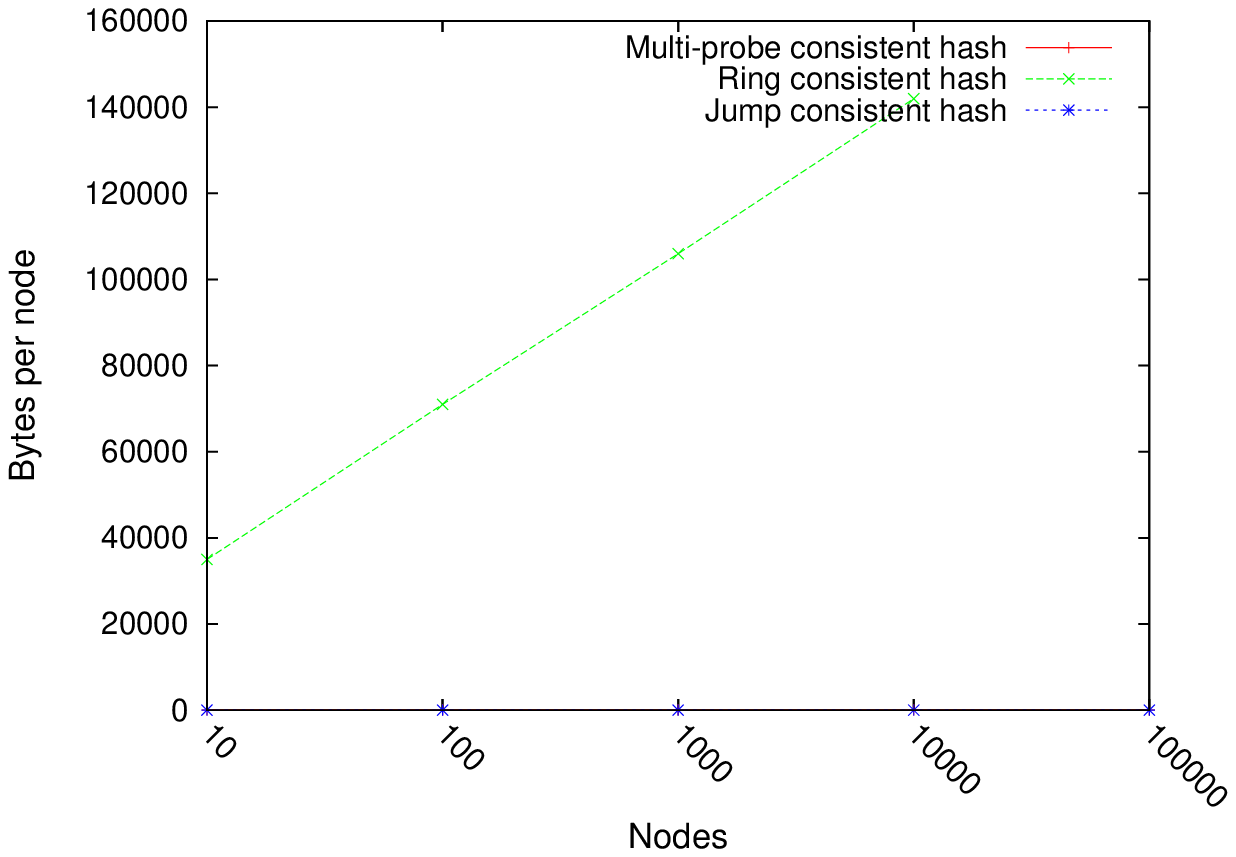}
  \label{figure:memory}
  \caption{Memory (bytes) per node, $\varepsilon \leq 0.05$}
\end{figure}

\begin{figure}
  \centering
  \includegraphics{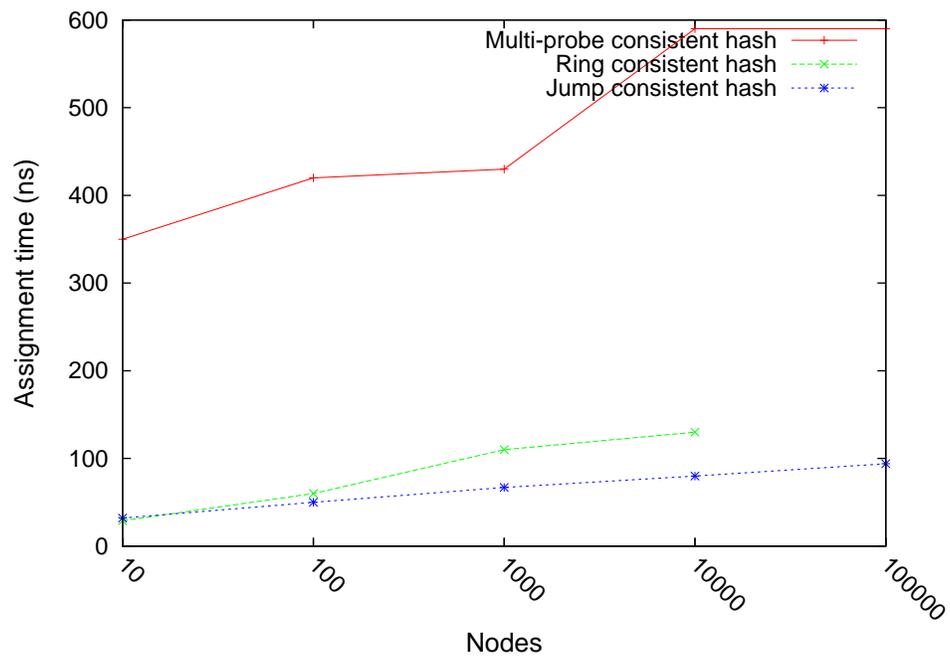}
  \label{figure:time}
  \caption{Assignment time (ns) per key, $\varepsilon \leq 0.05$}
\end{figure}

\begin{figure}
  \centering
  \includegraphics{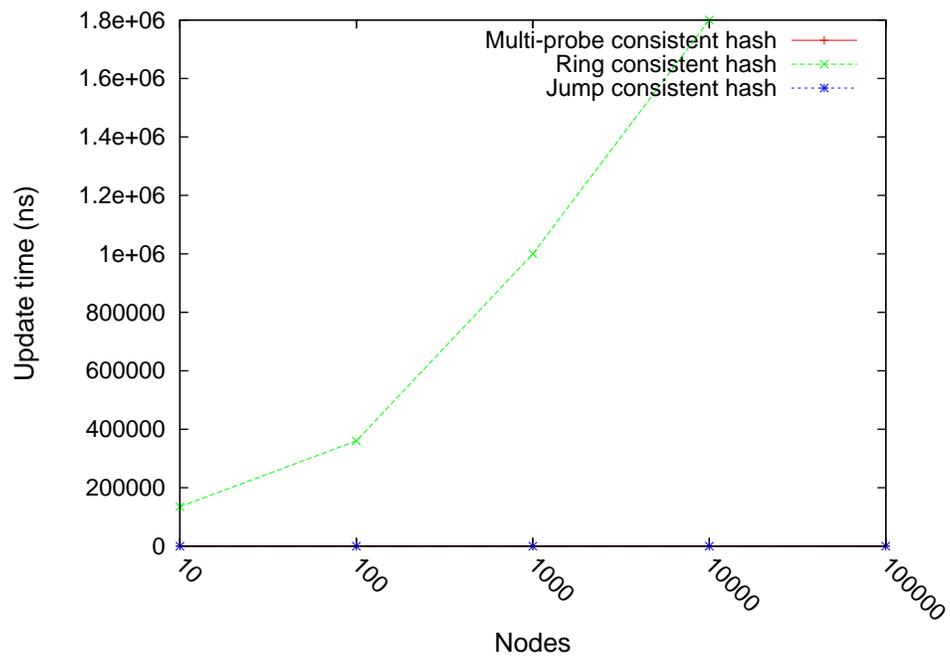}
  \label{figure:update}
  \caption{Update time (ns), $\varepsilon \leq 0.05$}
\end{figure}

\end{document}